\documentclass[]{interact}

\usepackage{algorithm}
\usepackage{algpseudocode} 
\usepackage{epstopdf}
\usepackage[caption=false]{subfig}
\usepackage{adjustbox}
\usepackage{rotating}
\usepackage{booktabs}
\usepackage{makecell}
\usepackage{tikz}
\usepackage{xcolor}
\usepackage{threeparttable}
\usepackage{url}
\usepackage{graphicx}
\usetikzlibrary{arrows.meta, patterns, positioning, calc}
\usepackage[authoryear,round]{natbib}
\bibpunct[, ]{(}{)}{;}{a}{,}{,}

\renewcommand\bibfont{\fontsize{10}{12}\selectfont}

\theoremstyle{plain}
\newtheorem{theorem}{Theorem}[section]

\newtheorem{proposition}[theorem]{Proposition}

\theoremstyle{definition}

\theoremstyle{remark}

\begin{document}


\title{Exact SAT and Constraint Programming for Job Shop Scheduling with Time-Varying Peak Power Constraints}

\author{
\name{Huy Tuan Nguyen\textsuperscript{a}, Duc Trung Kim Nguyen\textsuperscript{a}, Khanh To Van\textsuperscript{a}*\thanks{*Corresponding author. Email: khanhtv@vnu.edu.vn}}
\affil{\textsuperscript{a}Faculty of Information Technology, \\VNU University of Engineering and Technology, Hanoi, Vietnam}
}

\maketitle

\begin{abstract}
The Job Shop Scheduling Problem with Power Requirements (JSPPR) extends the classical job shop scheduling problem by imposing time-varying limits on instantaneous power consumption. Previous studies have used a mixed-integer linear programming formulation and the GRASP $\times$ ELS metaheuristic, but no SAT-based exact approach or constraint programming model has been reported. This paper develops the first exact SAT and constraint programming (CP) formulations for the JSPPR. On the 35 published benchmark instances, both SAT and CP prove global optimality for all instances and obtain identical optimal makespans, substantially improving upon the best previously reported results. They also establish four improved makespan values over the GRASP $\times$ ELS results reported in the original study. CP proves optimality faster than SAT, while both exact approaches substantially improve the optimality coverage of the MILP formulations, which prove optimality on only 6 and 10 instances using CPLEX and Gurobi, respectively. The certified optimal solutions also reveal inconsistencies in several previously reported benchmark results, including makespans below the proven optimum. We provide corrected optimal makespans and a complete set of certified optimal results for the JSPPR benchmark, establishing a reliable reference for future studies.
\end{abstract}


\section{Introduction}
\label{sec:introduction}
Manufacturing scheduling plays a central role in improving productivity, resource utilization, and production efficiency. The Job Shop Scheduling Problem (JSSP) is one of the most extensively studied scheduling problems because of its practical relevance and computational complexity. In modern manufacturing, however, scheduling decisions increasingly need to account for energy-related operational constraints in addition to conventional precedence and machine capacity requirements. Increasing energy costs, environmental concerns, and the integration of renewable energy sources have encouraged manufacturers to incorporate energy-related considerations into production scheduling \citep{duflou2012towards}. Beyond minimizing total energy consumption or electricity cost, many industrial facilities must operate under contractual or physical limits on instantaneous electrical power consumption. Exceeding these limits may result in financial penalties, equipment overload, or grid instability \citep{kemmoe2017job}. Consequently, production schedules must satisfy not only traditional precedence and machine capacity constraints but also ensure that the total power demand remains below the available power threshold at any point in time.

The Job Shop Scheduling Problem with Power Requirements (JSPPR), introduced by \citet{kemmoe2017job}, extends the classical JSSP by associating power requirements with operations and imposing a time-dependent upper bound on total instantaneous power consumption. These constraints introduce additional coupling among operations: operations that are independent in the classical JSSP may become mutually constrained when executed simultaneously because of their combined power demand. The resulting problem seeks a minimum makespan schedule satisfying precedence, machine capacity, and time-varying power constraints.

To date, research on the JSPPR has been limited to the MILP formulation and the GRASP $\times$ ELS metaheuristic of \citet{kemmoe2017job}. The MILP formulation was able to prove optimality for only a limited number of benchmark instances, while exhibiting relatively large optimality gaps compared with the high-quality solutions obtained by GRASP $\times$ ELS. In contrast, SAT and CP have proven effective as exact optimization paradigms for a wide range of combinatorial scheduling problems. Nevertheless, to the best of our knowledge, the JSPPR has not previously been investigated using either SAT or CP.

This paper addresses this gap by developing exact SAT and CP approaches for the JSPPR. We introduce the first SAT-based approach, using an order-based encoding of operation start times and a pseudo-Boolean encoding of the instantaneous power constraints. We also develop a CP formulation by reformulating the MILP model of \citet{kemmoe2017job} using interval variables and global constraints. 
We evaluate both approaches on the complete set of 35 published JSPPR benchmark instances and provide a systematic exact assessment of their computational performance and the existing benchmark results.

The main contributions of this paper are summarized as follows:
\begin{itemize}
\item First, we develop exact SAT and CP approaches for the JSPPR, introducing the first SAT-based approach with an order-based encoding and PB constraints, and a CP formulation using interval variables and global constraints.

\item Second, experiments on the 35 published JSPPR instances show that both SAT and CP prove global optimality for all instances and establish four new bounds beyond those reported by the GRASP $\times$ ELS method of \citet{kemmoe2017job}. CP is substantially faster at proving optimality, while SAT substantially improves the optimality coverage of the reimplemented MILP formulations.

\item Third, the certified optimal solutions reveal inconsistencies in several previously reported benchmark results. We identify these discrepancies and provide corrected optimal makespans, establishing a reliable reference for future computational studies of the JSPPR.

\end{itemize}

The remainder of this paper is organized as follows. Section \ref{sec:jsppr} presents the JSPPR and reviews related work. Section \ref{sec:preprocessing} describes the preprocessing procedures, while Sections \ref{sec:cp-formulation} and \ref{sec:sat-formulation} present the CP and SAT formulations, respectively. Section \ref{sec:computational-experiments} reports and discusses the computational results. Finally, Section \ref{sec:conclusion} concludes the paper and outlines future research directions.

\section{JSPPR: Problem Definition and Related Work}
\label{sec:jsppr}
\subsection{Problem Definition}
\label{sec:problem-definition}
The Job Shop Scheduling Problem with Power Requirements (JSPPR) extends the classical Job Shop Scheduling Problem (JSSP) by associating each operation with a power requirement while imposing a limit on the total instantaneous power consumption. The objective is to determine a feasible schedule that minimizes the makespan while satisfying precedence, machine capacity, and power constraints.

The manufacturing system is modeled under the following assumptions. The set of jobs and their processing routes are known before scheduling begins. Each job consists of an ordered sequence of $m$ operations, visiting each machine in $\mathcal{M}$ exactly once in a predetermined order. Each operation must be processed on a predetermined machine for a fixed processing time. Operations are non-preemptive; once an operation starts processing, it must continue until completion without interruption. Each machine can process at most one operation at any time, and each operation can be processed by only one machine simultaneously. Machine failures are not considered during the scheduling horizon. The available power threshold is predetermined by the power supplier and cannot be modified. Following the problem definition of \citet{kemmoe2017job}, the power profile of each operation is represented by two consecutive sub-operations without delay: an initial peak-consumption phase followed by a nominal-consumption phase.

Let $\mathcal{J}=\{0,\ldots,n-1\}$ denote the set of jobs and $\mathcal{M}=\{M_0,\ldots,M_{m-1}\}$ denote the set of machines. Each job $j \in \mathcal{J}$ consists of an ordered sequence of $m$ operations $(op_{j,0},\ldots,op_{j,m-1})$, where operation $op_{j,i}$ is processed on machine $m_{j,i}$ for a duration $P_{j,i}$. The operation enters its peak-power phase immediately when its execution begins, and this phase lasts for $a_{j,i}$ time units, where $a_{j,i}\le P_{j,i}$. The remaining $P_{j,i}-a_{j,i}$ time units constitute the nominal-power phase. Throughout its execution, operation $op_{j,i}$ consumes a nominal power of $W_{j,i}$. During the peak-power phase, its total consumption reaches the peak level $WS_{j,i}$, requiring an additional power draw of $W'_{j,i} = WS_{j,i} - W_{j,i}$. The maximum available power at time $t$ is denoted by the time-dependent power threshold $PT_t$. The notation used throughout this paper is summarized in Table~\ref{tab:notation}.

\begin{table}[htbp]
\centering
\caption{Notation used in this paper}
\label{tab:notation}
\begin{tabular}{ll}
\toprule
Symbol & Meaning \\
\midrule
$\mathcal{J}$ & Set of jobs \\
$\mathcal{M}$ & Set of machines \\
$n$ & Total number of jobs \\
$m$ & Total number of machines \\
$op_{j,i}$ & Operation $i$ of job $j$  \\
$P_{j,i}$ & Processing time of operation $op_{j,i}$ \\
$a_{j,i}$ & Peak-phase duration of operation $op_{j,i}$ \\
$W_{j,i}$ & Nominal power consumption of operation $op_{j,i}$ \\
$W'_{j,i}$ & Additional power draw of operation $op_{j,i}$ during peak phase \\
$PT_t$ & Available power threshold at time instant $t$ \\
\bottomrule
\end{tabular}
\end{table}

A feasible schedule must satisfy several constraints. First, precedence constraints ensure that the operations of each job are executed in their prescribed order. Specifically, for every job $j \in \mathcal{J}$ and every $i=0,\ldots,m-2$, operation $op_{j,i+1}$ cannot begin before $op_{j,i}$ has completed. Second, machine capacity constraints prevent two operations assigned to the same machine from overlapping in time. Third, the peak-power and nominal-power phases of each operation must proceed consecutively without interruption. Finally, for every time instant $t$, the total power consumed by all simultaneously executing operations, including both nominal power consumption $W_{j,i}$ and any additional power draw $W'_{j,i}$, must not exceed the available power threshold $PT_t$. The objective is to determine a feasible schedule that minimizes the makespan, denoted by $C_{\max}$, while satisfying all precedence, machine capacity, continuity, and power constraints. 

To illustrate the mathematical formulation and the behavior of the power constraints, we consider a small example comprising two jobs and two machines ($n = 2, m = 2$), adapted from \citet{kemmoe2017job}. Let $\mathcal{J} = \{0, 1\}$ and $\mathcal{M} = \{M_0, M_1\}$. Each job consists of two sequentially ordered operations, each consisting of a peak-power phase followed by a nominal-power phase. The routing, processing durations, and power profile parameters for each operation are summarized in Table~\ref{tab:jsppr_instance}.

\begin{table}[htbp]
\centering
\caption{Operational and Power Parameters of the Illustrative $2 \times 2$ JSPPR Instance}
\label{tab:jsppr_instance}
\begin{adjustbox}{width=\linewidth}
\begin{tabular}{ccccccc}
\hline
\textbf{Job} & \textbf{Operation} & \textbf{Machine} & \textbf{Total Dur.} & \textbf{Peak Dur.} & \textbf{Peak Power} & \textbf{Nominal Power} \\
$j$ & $op_{j,i}$ & $m_{j,i}$ & $P_{j,i}$ & $a_{j,i}$ & $WS_{j,i}$ (kW) & $W_{j,i}$ (kW) \\ \hline
$0$ & $op_{0,0}$ & $M_0$ & 89 & 8  & 38 & 17 \\
    & $op_{0,1}$ & $M_1$ & 35 & 10 & 37 & 21 \\ \hline
$1$ & $op_{1,0}$ & $M_1$ & 30 & 17 & 44 & 20 \\
    & $op_{1,1}$ & $M_0$ & 54 & 18 & 47 & 15 \\ \hline
\end{tabular}
\end{adjustbox}
\end{table}

The available power threshold imposed by the supplier varies dynamically over the scheduling horizon. This time-dependent power threshold $PT_t$ is defined piecewise as follows:

\begin{equation*}
PT_t = 
\begin{cases}
84\text{ kW}, & 0 \le t < 44, \\
40\text{ kW}, & 44 \le t < 89, \\
5\text{ kW},  & 89 \le t < 109, \\
15\text{ kW}, & 109 \le t < 133, \\
56\text{ kW}, & 133 \le t < 196, \\
0\text{ kW},  & 196 \le t < 220, \\
75\text{ kW}, & t \ge 220.
\end{cases}
\end{equation*}

The resulting schedule and its corresponding instantaneous power consumption profile are illustrated in Figure~\ref{fig:jsppr_gantt_power}. As illustrated, operations $op_{0,0}$ on $M_0$ and $op_{1,0}$ on $M_1$ start concurrently at $t = 0$, requiring a combined initial peak power of $82\text{ kW}$ which respects the threshold of $PT_t = 84\text{ kW}$. Operation $op_{1,0}$ finishes at $t = 30$, while $op_{0,0}$ drops to its nominal level at $t = 8$ and completes at $t = 89$. Although machine $M_0$ becomes available at $t = 89$, $op_{1,1}$ cannot start due to reduced power thresholds ($5\text{ kW}$ and $15\text{ kW}$), resulting in an unavoidable idle window across $[89, 133)$.

Once the threshold increases to $56\text{ kW}$ at $t = 133$, $op_{1,1}$ begins processing on $M_0$. When its peak phase finishes at $t = 151$, sufficient capacity is released to concurrently start $op_{0,1}$ on $M_1$ without exceeding $PT_t$. Operations $op_{0,1}$ and $op_{1,1}$ finish at $t = 186$ and $t = 187$, respectively, concluding all processing prior to the maintenance blackout at $t = 196$ and achieving the optimal makespan $C_{\max} = 187$.

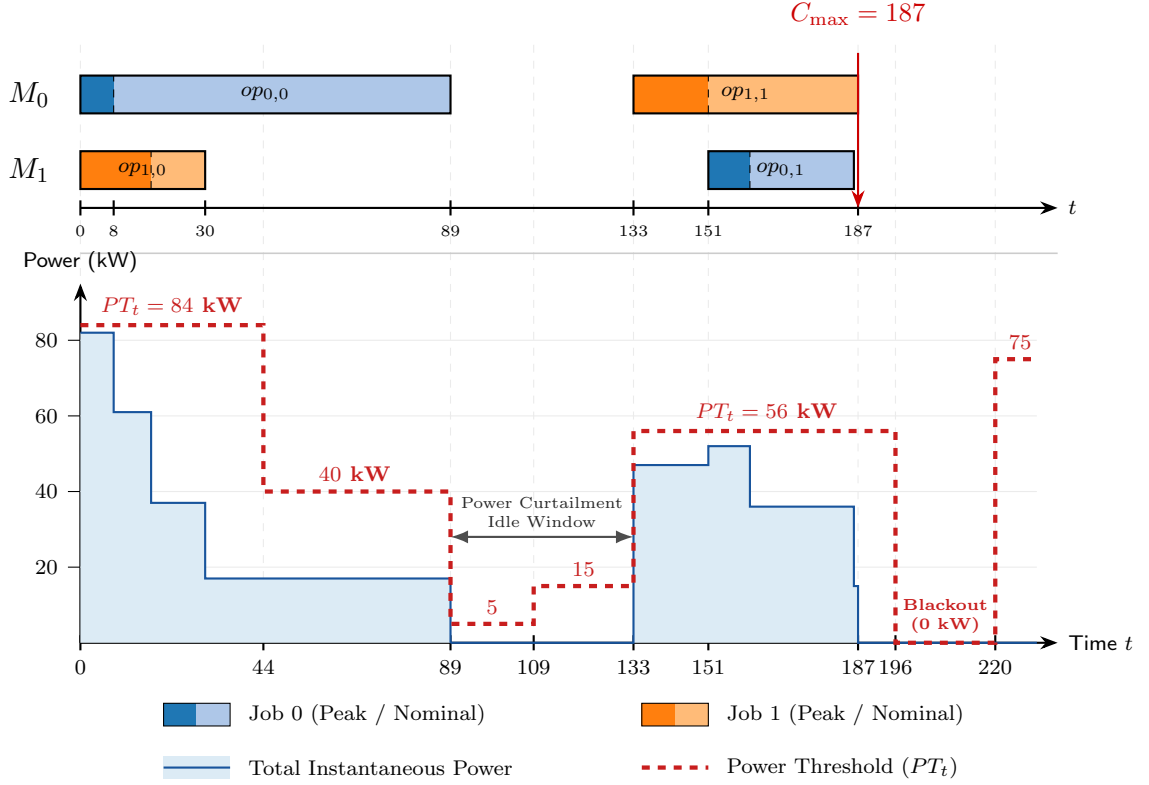
\begin{figure}[t]
\centering
\begin{tikzpicture}[
    x=0.055cm, y=0.05cm,
    font=\sffamily\footnotesize,
    >=Stealth
]
    \definecolor{job0peak}{RGB}{31, 119, 180}
    \definecolor{job0nom}{RGB}{174, 199, 232}
    \definecolor{job1peak}{RGB}{255, 127, 14}
    \definecolor{job1nom}{RGB}{255, 187, 120}
    \definecolor{powerfill}{RGB}{220, 235, 245}
    \definecolor{powerline}{RGB}{24, 84, 156}
    \definecolor{threshcolor}{RGB}{200, 30, 30}
    \definecolor{gridgray}{RGB}{235, 235, 235}

    \foreach \t in {0, 44, 89, 109, 133, 151, 187, 196, 220} {
        \draw[gridgray, dashed, line width=0.4pt] (\t, 0) -- (\t, 160);
    }

    \draw[->, thick] (0, 115) -- (235, 115) node[right] {$t$};
    \node[left, font=\bfseries] at (-5, 145) {$M_0$};
    \node[left, font=\bfseries] at (-5, 125) {$M_1$};

    \foreach \t in {0, 8, 30, 89, 133, 151, 187} {
        \draw[thick] (\t, 113) -- (\t, 117);
        \node[below, font=\tiny] at (\t, 113) {\t};
    }

    \fill[job0peak] (0, 140) rectangle (8, 150);
    \fill[job0nom]  (8, 140) rectangle (89, 150);
    \draw[thick, black] (0, 140) rectangle (89, 150);
    \draw[dashed, black] (8, 140) -- (8, 150);
    \node[font=\scriptsize] at (44.5, 145) {$op_{0,0}$};

    \fill[job1peak] (133, 140) rectangle (151, 150);
    \fill[job1nom]  (151, 140) rectangle (187, 150);
    \draw[thick, black] (133, 140) rectangle (187, 150);
    \draw[dashed, black] (151, 140) -- (151, 150);
    \node[font=\scriptsize] at (160, 145) {$op_{1,1}$};

    \fill[job1peak] (0, 120) rectangle (17, 130);
    \fill[job1nom]  (17, 120) rectangle (30, 130);
    \draw[thick, black] (0, 120) rectangle (30, 130);
    \draw[dashed, black] (17, 120) -- (17, 130);
    \node[font=\scriptsize] at (15, 125) {$op_{1,0}$};

    \fill[job0peak] (151, 120) rectangle (161, 130);
    \fill[job0nom]  (161, 120) rectangle (186, 130);
    \draw[thick, black] (151, 120) rectangle (186, 130);
    \draw[dashed, black] (161, 120) -- (161, 130);
    \node[font=\scriptsize] at (168.5, 125) {$op_{0,1}$};

    \draw[red!80!black, thick, -{Stealth}] (187, 156) -- (187, 115) 
        node[pos=-0.25, font=\bfseries\small] {$C_{\max} = 187$};

    \draw[gray!40, line width=0.6pt] (0, 103) -- (235, 103);

    \draw[->, thick] (0, 0) -- (235, 0) node[right] {Time $t$};
    \draw[->, thick] (0, 0) -- (0, 95) node[above] {Power (kW)};

    \foreach \p in {20, 40, 60, 80} {
        \draw (0, \p) -- (-3, \p) node[left, font=\scriptsize] {\p};
        \draw[gridgray, line width=0.4pt] (0, \p) -- (230, \p);
    }

    \foreach \t in {0, 44, 89, 109, 133, 151, 187, 196, 220} {
        \draw[thick] (\t, -2) -- (\t, 2);
        \node[below, font=\scriptsize] at (\t, -2) {\t};
    }

    \fill[powerfill] 
        (0,0) -- 
        (0,82) -- (8,82) -- (8,61) -- (17,61) -- (17,37) -- (30,37) -- (30,17) -- (89,17) -- (89,0) --
        (133,0) -- (133,47) -- (151,47) -- (151,52) -- (161,52) -- (161,36) -- (186,36) -- (186,15) -- (187,15) -- (187,0) -- cycle;

    \draw[powerline, thick] 
        (0,82) -- (8,82) -- (8,61) -- (17,61) -- (17,37) -- (30,37) -- (30,17) -- (89,17) -- (89,0) 
        (89,0) -- (133,0) -- (133,47) -- (151,47) -- (151,52) -- (161,52) -- (161,36) -- (186,36) -- (186,15) -- (187,15) -- (187,0) -- (230,0);

    \draw[threshcolor, ultra thick, dashed] 
        (0,84) -- (44,84) 
        -- (44,40) -- (89,40) 
        -- (89,5) -- (109,5) 
        -- (109,15) -- (133,15) 
        -- (133,56) -- (196,56) 
        -- (196,0) -- (220,0) 
        -- (220,75) -- (230,75);

    \node[threshcolor, above, font=\scriptsize\bfseries] at (22, 84) {$PT_t = 84\text{ kW}$};
    \node[threshcolor, above, font=\scriptsize\bfseries] at (66, 40) {$40\text{ kW}$};
    \node[threshcolor, above, font=\scriptsize\bfseries] at (99, 5) {$5$};
    \node[threshcolor, above, font=\scriptsize\bfseries] at (121, 15) {$15$};
    \node[threshcolor, above, font=\scriptsize\bfseries] at (165, 56) {$PT_t = 56\text{ kW}$};
    \node[threshcolor, above, font=\tiny\bfseries, align=center] at (208, 0) {Blackout\\(0 kW)};
    \node[threshcolor, above, font=\scriptsize\bfseries] at (226, 75) {$75$};

    \draw[<->, >=Latex, thick, gray!60!black] (89, 28) -- (133, 28);
    \node[above, font=\tiny, text=gray!30!black, align=center] at (111, 28) {Power Curtailment\\Idle Window};

    \begin{scope}[shift={(20, -22)}]
        \fill[job0peak] (0, 0) rectangle (8, 6);
        \fill[job0nom]  (8, 0) rectangle (16, 6);
        \draw[black, thin] (0, 0) rectangle (16, 6);
        \node[right, font=\scriptsize] at (18, 3) {Job $0$ (Peak / Nominal)};

        \fill[job1peak] (115, 0) rectangle (123, 6);
        \fill[job1nom]  (123, 0) rectangle (131, 6);
        \draw[black, thin] (115, 0) rectangle (131, 6);
        \node[right, font=\scriptsize] at (133, 3) {Job $1$ (Peak / Nominal)};

        \fill[powerfill] (0, -14) rectangle (16, -8);
        \draw[powerline, thick] (0, -11) -- (16, -11);
        \node[right, font=\scriptsize] at (18, -11) {Total Instantaneous Power};

        \draw[threshcolor, ultra thick, dashed] (115, -11) -- (131, -11);
        \node[right, font=\scriptsize] at (133, -11) {Power Threshold ($PT_t$)};
    \end{scope}

\end{tikzpicture}
\caption{Visualization of the optimal JSPPR schedule ($C_{\max} = 187$) and the corresponding power consumption profile under the time-varying power threshold $PT_t$ (dashed red line).}
\label{fig:jsppr_gantt_power}
\end{figure}

\subsection{Related Work}
\label{sec:related-work}

The Job Shop Scheduling Problem (JSSP) \citep{graham1966bounds} is a classical NP-hard \citep{lenstra1979computational} scheduling problem with numerous applications in manufacturing. Over the past decades, a wide range of solution methods have been proposed, including constructive heuristics and metaheuristics \citep{colorni1994ant, wisittipanich2012two, mahmud2022switching}, as well as exact approaches based on Mixed Integer Linear Programming (MILP) \citep{bowman1959schedule, ku2016mixed}, Constraint Programming (CP) \citep{da2022industrial, abreu2026mixed}, and Boolean Satisfiability (SAT) \citep{koshimura2010solving, huang2018efficient}. Exact methods provide optimality guarantees, whereas heuristic and metaheuristic methods generally aim to obtain high-quality feasible schedules within reasonable computational time.

Among exact approaches, CP and SAT have demonstrated competitive performance for the classical JSSP. CP models naturally capture precedence and machine capacity constraints through interval variables and specialized global constraints such as $\operatorname*{noOverlap}$. Specifically, \citet{da2022industrial} evaluated state-of-the-art CP solvers on large-scale industrial benchmarks containing up to one million operations, demonstrating the high scalability and solving efficiency of CP formulations under heavy workloads. Furthermore, \citet{abreu2026mixed} conducted an extensive comparison between commercial and open-source CP and MILP solvers on standard benchmark sets, showing that CP models consistently outperform MILP formulations, particularly on larger instances. SAT-based approaches, in contrast, encode scheduling decisions and resource conflicts as Boolean variables and propositional clauses. \citet{koshimura2010solving} applied the SAT encoding of \citet{crawford1994experimental} to solve historically open JSSP benchmark instances, while \citet{huang2018efficient} proposed a more compact encoding by tightening the feasible scheduling horizon of each operation, thereby reducing the number of Boolean variables and clauses. These developments highlight the strong potential of both CP and SAT paradigms for exact JSSP solving and motivate their extension to scheduling problems with additional resource constraints.

Energy-aware scheduling extends the classical JSSP by incorporating either energy-related objectives or additional operational constraints. One research direction minimizes cumulative energy consumption, sometimes jointly with makespan. For example, \citet{mouzon2007operational} considered machine operating and idle states, while \citet{may2015multi} proposed a multi-objective genetic algorithm for minimizing both energy consumption and makespan in job shop scheduling. Another direction minimizes electricity cost under Time-of-Use (TOU) tariffs by shifting production loads across different pricing periods. Representative studies include the hybrid flow-shop scheduling approach of \citet{luo2013hybrid} and the integrated job shop scheduling model of \citet{sanogo2026integrated}. These studies primarily treat energy as a cumulative objective or as a time-dependent operating cost.

A different class of scheduling problems imposes explicit limits on instantaneous power consumption. Rather than minimizing cumulative energy usage, these problems constrain the total power demand of simultaneously executing operations so that it never exceeds the available power capacity. Consequently, a schedule with low total energy consumption may still be infeasible if excessive power is required at a particular time. \citet{fang2013flow} investigated a permutation flow shop scheduling problem with peak-power constraints, while \citet{liu2018scheduling} studied peak-power control for two interfering job sets processed on parallel machines. These studies highlight that instantaneous power constraints introduce a fundamentally different scheduling requirement from cumulative energy optimization.

The Job Shop Scheduling Problem with Power Requirements (JSPPR), introduced by \citet{kemmoe2017job}, belongs to this class of scheduling problems with instantaneous power constraints. In the JSPPR, each operation is associated with a power requirement, and the total instantaneous power consumption throughout the schedule must not exceed a specified threshold. To solve the problem, \citet{kemmoe2017job} proposed an MILP formulation together with a metaheuristic approach called GRASP $\times$ ELS. The authors also identified the development of a CP formulation as a direction for future research. However, to the best of our knowledge, no SAT formulation or CP model for the JSPPR has subsequently been reported. 

These observations reveal a methodological gap in the existing literature on the JSPPR: despite the availability of a mathematical programming model and a metaheuristic approach, the problem has not yet been investigated using SAT or CP as alternative exact optimization paradigms. This paper addresses this gap by introducing a SAT formulation for the JSPPR and a CP formulation derived from the existing MILP model, thereby providing alternative exact approaches for the problem and enabling a systematic comparison of different solution paradigms.

\section{Preprocessing for Exact Approaches}
\label{sec:preprocessing}
\subsection{Randomized Greedy Upper Bound Computation}
\label{subsec:upper-bound}

Establishing a finite upper bound on the makespan is essential for both SAT and CP approaches, as it defines the scheduling horizon and restricts the variable domains and search space. For the JSPPR, obtaining such a bound is non-trivial because an upper bound based on total processing time is insufficient when the power threshold $PT_t$ varies over time.

We therefore use a randomized greedy construction heuristic to obtain an instance-specific upper bound. The heuristic incrementally constructs a feasible schedule by repeatedly selecting an unscheduled job uniformly at random and assigning its next operation to the earliest start time satisfying the precedence, machine capacity, and power constraints. If no feasible start time exists, the current construction is abandoned. The procedure is repeated within a fixed time limit, and the feasible schedule with the smallest makespan is retained.

Let $UB$ denote the makespan of the best schedule found. The scheduling horizon used in the remainder of this paper is then set to
\begin{equation}
H = UB.
\end{equation}
This horizon is subsequently used in the start-time domain reduction and the SAT and CP formulations. Algorithm~\ref{alg:upper-bound} summarizes
the randomized greedy upper-bound computation.

\begin{algorithm}[htbp]
\caption{Randomized Greedy Upper Bound Generation}
\label{alg:upper-bound}
\begin{algorithmic}[1]
\scriptsize
\Require JSPPR instance $I$ with horizon $H$, time limit $T$
\Ensure Upper bound $UB$ and corresponding schedule

\State $UB \gets \infty$, $bestSchedule \gets \emptyset$

\While{elapsed time $<T$}
    \State Initialize an empty schedule
    \While{unscheduled operations remain}
        \State Randomly select a job $j$ with an unscheduled operation
        \State Let $o$ be the next unscheduled operation of $j$
        \State Determine the earliest feasible start time $t$ for $o$
        \If{no feasible start time exists within $H$}
            \State Discard the current construction
            \State \textbf{break}
        \EndIf
        \State Schedule $o$ at $t$ and update resource availability
    \EndWhile

    \If{a complete feasible schedule was constructed}
        \State $C_{\max} \gets$ makespan of the schedule
        \If{$C_{\max}<UB$}
            \State $UB \gets C_{\max}$
            \State $bestSchedule \gets$ current schedule
        \EndIf
    \EndIf
\EndWhile

\State \Return $(UB,bestSchedule)$
\end{algorithmic}
\end{algorithm}

\subsection{Start Time Domain Reduction}
\label{subsec:domain-reduction}
Prior to constructing the scheduling models, a domain reduction procedure is applied to tighten the feasible start-time domain of each operation. For each operation $op_{j,i}$, it computes an earliest feasible start time $ES_{j,i}$ and a latest feasible start time $LS_{j,i}$. The bounds are derived solely from the precedence constraints within each job and the current scheduling horizon $H$. Machine capacity and power constraints are enforced directly by the solver formulations and are therefore not considered during this preprocessing step.

Assuming all jobs are available at time $0$, the earliest start time of the initial operation is $ES_{j,0} = 0$ for all $j \in \mathcal{J}$. Because operations within each job must be executed sequentially, the earliest start time for each subsequent operation $op_{j,i}$ is defined as:
\begin{equation}
ES_{j,i} = \sum_{k=0}^{i-1} P_{j,k} \quad \forall j \in \mathcal{J}, \; \forall i \in \{1, \ldots, m-1\}.
\end{equation}

To determine the latest feasible start time with respect to a given scheduling horizon $H$, let $rem_{j,i}$ denote the remaining processing time required from operation $op_{j,i}$ to the completion of job $j$:
\begin{equation}
rem_{j,i} = \sum_{k=i}^{m-1} P_{j,k} \quad \forall j \in \mathcal{J}, \; \forall i \in \{0, \ldots, m-1\}.
\end{equation}
The latest feasible start time of operation $op_{j,i}$ is then bounded by:
\begin{equation}
LS_{j,i} = H - rem_{j,i} \quad \forall j \in \mathcal{J}, \; \forall i \in \{0, \ldots, m-1\}.
\end{equation}

Thus, each operation is assigned the reduced start-time domain $[ES_{j,i},LS_{j,i}]$.
If $ES_{j,i}>LS_{j,i}$ for any operation, the instance is identified as infeasible under the current scheduling horizon, and model construction is skipped. Otherwise, the resulting domains are used to restrict the start-time variables in both the CP and SAT formulations.

\section{Constraint Programming Formulation}
\label{sec:cp-formulation}
\setcounter{equation}{0}
\renewcommand{\theequation}{CP-\arabic{equation}}

The CP formulation presented in this section is derived from the MILP formulation
proposed by \citet{kemmoe2017job} by reformulating its scheduling relationships
and resource constraints using CP interval variables and global constraints.
This reformulation preserves the underlying problem structure while exploiting
CP-specific constructs to represent temporal and resource interactions more
naturally. In particular, each operation is represented by a mandatory master
interval, its power profile is decomposed into phase intervals, and the resulting
intervals are coupled through job-precedence, machine capacity, and power resource
constraints. The formulation then minimizes the makespan of the resulting
feasible schedule.

Each operation $op_{j,i}$ is modeled as a mandatory master interval variable
$I_{j,i}$ covering its entire execution, with fixed length $P_{j,i}$ and a start-time
domain restricted to $[ES_{j,i}, LS_{j,i}]$ obtained during preprocessing:
\begin{equation}
\begin{aligned}
& ES_{j,i} \le \operatorname{startOf}(I_{j,i}) \le LS_{j,i}, \\
& \operatorname{length}(I_{j,i}) = P_{j,i}
\end{aligned}
\qquad \forall j \in \mathcal{J},\; \forall i \in \{0, \dots, m-1\}.
\end{equation}

The master interval also determines the power consumption of the operation through a sequence of phases. Let $\Phi_{j,i} \subseteq \{\mathsf{p}, \mathsf{n}\}$ denote the set of phases present in operation $op_{j,i}$, where $\mathsf{p}$ represents the peak-consumption phase and $\mathsf{n}$ represents the nominal-consumption phase. A phase with zero duration is omitted from $\Phi_{j,i}$. For the peak phase $\mathsf{p}$, the duration is $d_{j,i,\mathsf{p}} = a_{j,i}$, the start offset is $o_{j,i,\mathsf{p}} = 0$, and the power demand is $r_{j,i,\mathsf{p}} = W_{j,i} + W'_{j,i}$. For the nominal phase $\mathsf{n}$, the duration is $d_{j,i,\mathsf{n}} = P_{j,i} - a_{j,i}$, the start offset is $o_{j,i,\mathsf{n}} = a_{j,i}$, and the power demand is $r_{j,i,\mathsf{n}} = W_{j,i}$. Each phase $\phi \in \Phi_{j,i}$ is represented by a mandatory interval variable $I_{j,i,\phi}$ with length $d_{j,i,\phi}$, positioned relative to the master interval by:
\begin{equation}
\begin{aligned}
& \operatorname{startOf}(I_{j,i,\phi}) = \operatorname{startOf}(I_{j,i}) + o_{j,i,\phi}, \\
& \operatorname{length}(I_{j,i,\phi}) = d_{j,i,\phi}
\end{aligned}
\qquad \forall j \in \mathcal{J},\; \forall i \in \{0, \dots, m-1\},\; \forall \phi \in \Phi_{j,i}.
\end{equation}
Since the master interval represents the complete execution of each operation, it directly expresses both precedence within each job and non-overlap among operations sharing a machine. For each job $j \in \mathcal{J}$, the operations must be processed in their prescribed order. The completion of an earlier operation must not occur after the start of any subsequent operation:
\begin{equation}
\operatorname{endOf}(I_{j,i}) \le \operatorname{startOf}(I_{j,k})
\qquad
\forall j \in \mathcal{J},\quad 0 \le i < k \le m-1.
\end{equation}
The pairwise form above also introduces precedence relations between non-consecutive operations. These relations are redundant in terms of the feasible solution set, since they are implied by the precedence constraints between consecutive operations. Nevertheless, explicitly posting these redundant constraints can improve propagation by directly linking the temporal bounds of non-consecutive operations, potentially reducing the search space and improving solving performance.

In addition to job precedence, operations assigned to the same machine cannot overlap in time. For each machine $M_q \in \mathcal{M}$, let $\mathcal{O}_q = \{(j,i) : m_{j,i} = M_q\}$ denote the set of operations processed on that machine. This requirement is enforced directly using the global no-overlap constraint:
\begin{equation}
\operatorname{noOverlap}
\left(
\{I_{j,i} : (j,i) \in \mathcal{O}_q\}
\right)
\qquad
\forall M_q \in \mathcal{M}.
\end{equation}

The maximum power consumption varies over the scheduling horizon. The power threshold profile is piecewise constant; let $0 = \tau_0 < \tau_1 < \dots < \tau_R = H$ denote its breakpoints. For each segment $[\tau_q, \tau_{q+1})$, let $\overline{PT}_q$ denote the maximum power threshold allowed throughout that segment.

Since the standard cumulative constraint provided by CP solvers assumes a single constant limit, it cannot directly represent the time-varying threshold $PT_t$. To address this limitation, we transform the time-varying threshold into a standard constant-limit cumulative constraint using fixed dummy intervals. Let the overall power threshold be set to the peak threshold across the entire horizon:
\begin{equation}
\overline{PT} = \max_{0 \le q < R} \overline{PT}_q.
\end{equation}

For each segment $q \in \{0, \dots, R-1\}$, we define a fixed interval $D_q = [\tau_q, \tau_{q+1})$ with a constant power consumption of $\overline{PT} - \overline{PT}_q$. Thus, during segment $q$, $D_q$ reserves the unusable power, leaving exactly $\overline{PT}_q$ available for the operations. The time-dependent power constraint is therefore enforced via a single cumulative constraint over the unified set of operational phases and dummy intervals:
\begin{equation}
\begin{aligned}
\operatorname{cumulative}\Bigl(
& \{(I_{j,i,\phi}, r_{j,i,\phi}) : j \in \mathcal{J},\, 0 \le i < m,\, \phi \in \Phi_{j,i}\} \\
& \cup \{(D_q, \overline{PT} - \overline{PT}_q) : 0 \le q < R\}, \\
& \overline{PT}
\Bigr).
\end{aligned}
\end{equation}
This ensures that the total power consumption of all active phases at any time $t$ strictly satisfies the original limit $PT_t$.

Finally, the makespan is defined as the completion time of the last operation:
\begin{equation}
C_{\max} = \max_{j \in \mathcal{J},\ 0 \le i < m} \operatorname{endOf}(I_{j,i}).
\end{equation}

The resulting CP formulation is
\begin{equation*}
\begin{aligned}
\min \quad & C_{\max} \\
\text{s.t.}\quad & \text{(CP-1)--(CP-7)}.
\end{aligned}
\end{equation*}

\section{SAT Formulation}
\label{sec:sat-formulation}
\setcounter{equation}{0}
\renewcommand{\theequation}{SAT-\arabic{equation}}

The proposed SAT formulation encodes the JSPPR as a discrete-time Boolean
satisfiability problem over the scheduling horizon $H$. Each operation is
represented by Boolean start-time variables, from which its execution phases
are derived. Based on this representation, job precedence, machine capacity,
and time-dependent power constraints are encoded as Boolean or pseudo-Boolean
constraints.

The formulation consists of two stages. First, a base encoding represents all
feasible schedules within the horizon $H$, with domain reduction incorporated
by restricting the feasible start times of each operation. Second, an
incremental optimization procedure imposes increasingly tighter makespan
bounds on the same base encoding until optimality is established.

\subsection{Encoding Variables and Constraints}
\label{subsec:sat-base-encoding}

The base SAT encoding is constructed over the scheduling horizon $H$ and
incorporates the reduced start-time domains obtained during preprocessing.
For each operation $o\in\mathcal{O}$, let $[ES_o,LS_o]$ denote its feasible
start-time domain. Hence, Boolean variables are introduced only for start
times that remain possible after domain reduction. For notational convenience,
we use $o=(j,i)$ for the $i$th operation of job $j$ and define
$p_o=P_{j,i}$, $a_o=a_{j,i}$, $W_o=W_{j,i}$, $W'_o=W'_{j,i}$, and
$m_o=m_{j,i}$. The scheduling instants are
$\mathcal{T}=\{0,\ldots,H-1\}$.

The temporal position of each operation is encoded using two Boolean variables: the order variable $X_{o,t}$, indicating whether operation $o$ starts at or after time $t$, and the exact-start variable $S_{o,t}$, indicating whether $o$ starts precisely at time $t$. Consequently, $X$ forms a cumulative representation of the start time, while $S$ pinpoints the exact transition. Leveraging this cumulative structure, domain reduction is straightforwardly achieved by fixing the order variables at the boundaries of the feasible start-time window:
\begin{equation}
    X_{o,ES_o}=\mathrm{true},
    \qquad
    X_{o,LS_o+1}=\mathrm{false},
    \qquad
    \forall o\in\mathcal{O}.
    \label{eq:sat-order-boundary}
\end{equation}
The order variables are monotone:
\begin{equation}
    \neg X_{o,t+1}\lor X_{o,t},
    \qquad
    \forall o\in\mathcal{O},\;
    t\in\{ES_o,\ldots,LS_o\}.
    \label{eq:sat-order-monotonicity}
\end{equation}
The exact-start variable is linked to the unique true-to-false transition of the order-variable sequence:
\begin{equation}
\begin{aligned}
    &\neg S_{o,t} \lor X_{o,t}, \\
    &\neg S_{o,t} \lor \neg X_{o,t+1}, \\
    &\neg X_{o,t} \lor X_{o,t+1} \lor S_{o,t},
\end{aligned}
\qquad
\forall o \in \mathcal{O},\;
t \in \{ES_o, \dots, LS_o\}.
\label{eq:sat-exact-start}
\end{equation}
Together with the boundary and monotonicity conditions, these clauses ensure that every operation is assigned exactly one start time in its reduced domain.

Once the start time of every operation has been established, temporal feasibility can be expressed directly from the selected start times. Precedence within each job is enforced between consecutive operations: if $(j,i)$ starts at time $t$, its successor $(j,i+1)$ must start no earlier than $t+p_{j,i}$:
\begin{equation}
\neg S_{(j,i),t} \lor X_{(j,i+1),\,t+p_{j,i}},
\qquad
\begin{aligned}
    & \forall j \in \mathcal{J},\; i \in \{0, \dots, m-2\}, \\
    & \forall t \in \{ES_{j,i}, \dots, LS_{j,i}\}.
\end{aligned}
\label{eq:sat-precedence}
\end{equation}
An order variable whose time index falls outside the representable horizon is interpreted as false. Hence, any start time that would require the successor to start beyond the horizon is ruled out by the precedence constraint.

Operations assigned to the same machine must also be non-overlapping. For any pair of distinct operations $\{o, o'\}$ with $m_o = m_{o'}$, if $o$ starts at time $t$, then $o'$ must either start after $o$ completes or finish before $o$ begins. In the cumulative representation, this mutual exclusion is encoded symmetrically for both directed pairs $(o, o')$ and $(o', o)$ across all feasible start times:
\begin{equation}
\neg S_{o,t} \lor X_{o',\,t+p_o} \lor \neg X_{o',\,t-p_{o'}+1},
\qquad
\begin{aligned}
    & \forall o,o' \in \mathcal{O},\; o \neq o',\; m_o = m_{o'}, \\
    & \forall t \in \{ES_o, \dots, LS_o\}.
\end{aligned}
\label{eq:sat-machine-no-overlap}
\end{equation}
Here, out-of-bound order variables are treated as fixed constants: $X_{o',t'}$ evaluates to $\text{true}$ for $t' \le ES_{o'}$ and $\text{false}$ for $t' > LS_{o'}$. By enforcing this clause symmetrically for every feasible start time, the encoding strictly forbids $o'$ from starting within the overlapping window $[t - p_{o'} + 1,\, t + p_o - 1]$, thereby eliminating any conflicting resource usage on the shared machine.

The start-time variables also determine the phase in which each operation consumes power. Two additional Boolean variables are introduced for this purpose: $P_{o,t}$ is true if operation $o$ is in its peak-consumption phase at time $t$, and $B_{o,t}$ is true if it is in its nominal-consumption phase. If operation $o$ starts at $s$, its peak phase occupies the first $a_o$ time instants and its nominal phase occupies the remaining $p_o-a_o$ instants:
\begin{equation}
\begin{aligned}
    &\neg S_{o,s} \lor P_{o,s+\delta},
    \quad && \forall \delta \in \{0, \dots, a_o-1\},\\
    &\neg S_{o,s} \lor B_{o,s+\delta},
    \quad && \forall \delta \in \{a_o, \dots, p_o-1\},
\end{aligned}
\qquad
\forall o \in \mathcal{O},\;
s \in \{ES_o, \dots, LS_o\}.
\label{eq:sat-phase-activation}
\end{equation}
No clause is generated for a phase of zero duration, and phase variables outside $\mathcal{T}$ are interpreted as false.

The phase variables provide the representation required to enforce the time-dependent power limit. At each time instant $t$, only operations whose corresponding phase may possibly cover $t$ need to participate in the power constraint. Define the candidate sets for each $t\in\mathcal{T}$:
\begin{equation}
\begin{aligned}
    \mathcal{O}^{\mathsf{p}}_t
    &=
    \left\{
        o\in\mathcal{O}:
        a_o>0,\;
        ES_o\le t\le LS_o+a_o-1,\;
        W_o+W'_o>0
    \right\},\\
    \mathcal{O}^{\mathsf{n}}_t
    &=
    \left\{
        o\in\mathcal{O}:
        p_o>a_o,\;
        ES_o+a_o\le t\le LS_o+p_o-1,\;
        W_o>0
    \right\}.
\end{aligned}
\label{eq:sat-power-candidate-sets}
\end{equation}
The instantaneous power demand is then obtained by summing the contributions of all active peak and nominal phases:
\begin{equation}
    \sum_{o\in\mathcal{O}^{\mathsf{p}}_t}
        (W_o+W'_o)P_{o,t}
    +
    \sum_{o\in\mathcal{O}^{\mathsf{n}}_t}
        W_oB_{o,t}
    \le PT_t,
    \qquad
    \forall t\in\mathcal{T}.
    \label{eq:sat-power-pb}
\end{equation}
This constraint is a pseudo-Boolean (PB) constraint, as its left-hand side is a weighted sum of Boolean variables. To represent these constraints in SAT, we employ the Binary Merger encoding \citep{manthey2014more}, an efficient encoding for PB constraints that converts the weighted-sum constraint into an equivalent CNF representation.

Consequently, the base encoding jointly specifies a feasible start time for
every operation, derives the corresponding execution phases, and enforces
precedence, machine non-overlap, and power feasibility over the scheduling
horizon. The resulting formula therefore characterizes the feasible-schedule
space, while optimization of the makespan is handled separately by the
incremental procedure described in the next subsection.

\subsection{Incremental Optimization}
\label{subsec:sat-optimization}

The base encoding characterizes feasible schedules within the horizon $H$, but does not minimize the makespan. Optimization is therefore performed incrementally by exploiting the stateful solving capabilities of modern CDCL SAT solvers \citep{biere2009handbook}. Rather than recreating SAT formulas and invoking independent solver instances from scratch, a single persistent solver instance maintains the base encoding $\mathcal{F}_{\mathrm{base}}$ alongside all learned conflict clauses and heuristic scores across successive calls.

For a makespan threshold $\widehat{C}_{\max}$, let $o_j^{\mathrm{last}}$ and $p_j^{\mathrm{last}}$ denote the last operation of job $j$ and its processing time, respectively. Because $X_{o,t}$ is true if and only if operation $o$ starts at or after time $t$, the condition $C_{\max}\leq\widehat{C}_{\max}$ is enforced by the conjunction of unit clauses
\begin{equation}
    \mathcal{A}(\widehat{C}_{\max})
    =
    \bigwedge_{j\in\mathcal{J}}
    \neg X_{o_j^{\mathrm{last}},
    \widehat{C}_{\max}-p_j^{\mathrm{last}}+1}.
    \label{eq:sat-makespan-assumption}
\end{equation}

The incremental solving loop is initialized by loading $\mathcal{F}_{\mathrm{base}}$ and the initial makespan bound $\mathcal{A}(H)$ into an active solver instance. In each iteration, the solver evaluates satisfiability while preserving its internal conflict graph. When a satisfiable assignment is found, it is decoded into a feasible schedule with makespan $C_{\max}$. The bound is then tightened to $C_{\max}-1$, and the new unit clauses $\mathcal{A}(C_{\max}-1)$ are appended directly to the current solver state without resetting the solver. Because makespan bounds are monotonically decreasing, previously added clauses remain logically valid. The optimization terminates when the solver returns unsatisfiable, proving that no feasible schedule with a strictly smaller makespan exists. If unsatisfiability occurs during the first check at $\widehat{C}_{\max}=H$, the instance is proven infeasible.

The incremental optimization procedure is outlined in Algorithm~\ref{alg:incremental-sat}.

\begin{algorithm}[htbp]
\caption{Incremental SAT Optimization}
\label{alg:incremental-sat}
\begin{algorithmic}[1]
\Require JSPPR instance with horizon $H$
\Ensure An optimal schedule or an infeasibility result
\State $\mathcal{S} \gets \text{InitializeSolver}(\mathcal{F}_{\mathrm{base}})$ \Comment{Instantiate persistent solver with base clauses}
\State $\widehat{C}_{\max} \gets H$
\State $bestSchedule \gets \emptyset$
\State $\mathcal{S}.\text{addClauses}(\mathcal{A}(\widehat{C}_{\max}))$
\While{$\mathcal{S}.\text{solve}() = \text{SAT}$} \Comment{Incremental check retaining learned clauses}
    \State $bestSchedule \gets \text{DecodeSchedule}(\mathcal{S}.\text{getModel}())$
    \State $C_{\max} \gets \text{ComputeMakespan}(bestSchedule)$
    \State $\widehat{C}_{\max} \gets C_{\max} - 1$
    \If{$\widehat{C}_{\max} < 0$}
        \State \textbf{break}
    \EndIf
    \State $\mathcal{S}.\text{addClauses}(\mathcal{A}(\widehat{C}_{\max}))$ \Comment{Append tighter bound to existing solver state}
\EndWhile
\If{$bestSchedule = \emptyset$}
    \State \Return infeasible
\Else
    \State \Return $bestSchedule$ \Comment{Incumbent is provably optimal}
\EndIf
\end{algorithmic}
\end{algorithm}

\subsection{Complexity Analysis}
\label{subsec:sat-complexity}

Let $p_{\max}=\max_{o\in\mathcal{O}}p_o$ and let
$w_{\max}=\max_{o\in\mathcal{O}}(W_o+W'_o)$. Since $[ES_o,LS_o]$ is the
feasible start-time domain for operation $o$, its size is
$D_o=|[ES_o,LS_o]|=LS_o-ES_o+1$. Thus, in the worst case, $D_o=O(H)$.

\begin{proposition}[Size of the SAT encoding]
\label{prop:base-encoding-complexity}
Under the worst-case domain size $D_o = O(H)$, the SAT encoding introduces
\begin{align}
    O\!\left(
        nmH\left(
            1 + \log^2(nm)\log w_{\max}
        \right)
    \right)
\end{align}
Boolean variables and
\begin{align}
    O\!\left(
        nmH\left(
            n + p_{\max}
            + \log^2(nm)\log w_{\max}
        \right)
    \right)
\end{align}
clauses.
\end{proposition}

\begin{proof}
The encoding uses four types of base Boolean variables over the feasible
start-time domains, yielding $O(nmH)$ variables. The boundary,
monotonicity, exact-start linking, and precedence constraints in
\eqref{eq:sat-order-boundary}--\eqref{eq:sat-precedence}
together contribute $O(nmH)$ clauses.

The machine non-overlap constraints in \eqref{eq:sat-machine-no-overlap},
including their symmetric counterparts, contribute $O(n^2mH)$ clauses. The
phase-activation constraints in \eqref{eq:sat-phase-activation} contribute
$O(nmHp_{\max})$ clauses.

For each time instant $t$, the Binary Merger encoding of the power
constraint in \eqref{eq:sat-power-pb} contributes
$O(nm\log^2(nm)\log w_{\max})$ variables and clauses
\citep{manthey2014more}. Since the scheduling horizon contains $H$ time
instants, the total contribution of the power constraints is
$O(nmH\log^2(nm)\log w_{\max})$ variables and clauses.

Summing the contributions of all variable and constraint families gives the stated bounds.
\end{proof}

\begin{proposition}[Complexity of the incremental optimization loop]
\label{prop:optimization-loop-complexity}
The incremental optimization loop performs at most $O(H)$ SAT calls, adds
$O(nH)$ assumption clauses in total, and introduces no new variables.
\end{proposition}

\begin{proof}
Each satisfiable iteration strictly decreases the integer threshold
$\widehat{C}_{\max}$ by at least 1, so there are at most $O(H)$ iterations.
Each iteration adds at most $n$ unit assumption clauses and no variables. Thus,
over all iterations, the incremental overhead is $O(nH)$ clauses and no
additional variables.
\end{proof}

\section{Computational Experiments}
\label{sec:computational-experiments}

\subsection{Benchmark Instances and Experimental Setup}

The benchmark instances used in this study were introduced by
\citet{kemmoe2017job} for evaluating exact and heuristic methods for the
Job Shop Scheduling Problem with Power Requirements. The instances used contain between four and ten jobs
and four machines. Their power-related data are generated randomly from
duration-dependent ranges. In particular, the basic and peak-power
parameters are generated using bounds between one-half and the full
duration of an operation, while the duration of the peak-consumption phase
is selected between zero and one-third of the operation duration. The
variable power threshold starts at 120\% of the minimum power required by
the complete system and is subsequently increased or decreased at selected
time points. The initial variation lies between $+25\%$ and $-25\%$;
after an increase, the following variation uses the range $+20\%$ to
$-30\%$. The probabilities of increasing and decreasing the threshold are
also adjusted to favor the larger absolute variation. A feasibility-preserving delay is applied before selecting the minimum required threshold.
These instances, identified as JSPPR instances, are publicly available in
the benchmark collection provided by the original authors\footnote{https://damienlamy.com/Works/Energy/JSPPR/VariableThreshold/}.

The computational experiments were conducted on a Google Cloud\footnote{https://cloud.google.com/}
\texttt{c4-highmem-4} virtual machine running Ubuntu 22.04.5 LTS.
The machine was equipped with an Intel Xeon Platinum 8581C processor at 2.30 GHz and 31 GB of RAM. The implementations were developed in Python 3.10.
The proposed SAT formulation was solved with CaDiCaL 1.9.5 \citep{BiereFazekasFleuryHeisinger-SAT-Competition-2020-solvers} through the PySAT library \citep{itk-sat24} version 1.9.dev4. The MILP formulation of
\citet{kemmoe2017job} was reimplemented and solved with CPLEX v22.2\footnote{https://www.ibm.com/products/ilog-cplex-optimization-studio}
and Gurobi v13.0.2\footnote{https://www.gurobi.com/}. The CP model was solved using CPLEX v22.2. The same preprocessing procedure described in Section~\ref{sec:preprocessing} was applied to all three formulations. A time limit of 3600 seconds was imposed on each instance for the SAT, MILP, and CP experiments. For comparison with existing results, we report the results of the exact CPLEX MILP and GRASP $\times$ ELS heuristic approaches as reported by \citet{kemmoe2017job}; these results were taken directly from the paper.

\subsection{Metrics}

The performance of the evaluated methods is assessed using the following metrics. Makespan is denoted by $C_{\max}$ and represents the completion time of the last operation in a schedule. Since the objective is makespan minimization, smaller values indicate better schedules.

The computational time, denoted by $t$ and reported in seconds, represents the total wall-clock time considered for a method on an instance, subject to the time limit specified in the experimental setup. For exact methods implemented by us, $t$ includes all computational steps, including reading the problem instance, computing the initial makespan upper bound using the fixed 5-second upper-bound generation procedure, performing domain reduction, and searching for and proving an optimal solution. Computational time is not considered for the CPLEX MILP and GRASP $\times$ ELS results reported by \citet{kemmoe2017job}, as these are historical results obtained under a different experimental setup.

The metric $\#\mathrm{OPTIMAL}$ denotes the number of benchmark instances for which a method proves optimality within the specified time limit. We calculate this metric exclusively for the methods implemented by us, as the results reported by \citet{kemmoe2017job} contain inconsistencies. The metric $\#\mathrm{BEST}$ denotes the number of instances for which a method achieves the best-known makespan among the compared methods. If multiple methods attain the same best makespan, each receives credit for that instance. Results from \citet{kemmoe2017job} identified as inconsistent in this work are excluded from the calculation.

\subsection{Results and Analysis}
Table~\ref{tab:aggregated-results} provides an aggregated overview of the computational performance and solution quality across all 35 JSPPR benchmark instances. Overall, the proposed SAT formulation and the CPLEX CP model show strong performance, with both achieving $\#\mathrm{OPTIMAL} = 35$ and attaining the best-known makespan on all benchmark instances ($\#\mathrm{BEST} = 35$). The GRASP $\times$ ELS metaheuristic also performs well, attaining the best makespan on 24 instances, although, as a metaheuristic, it does not provide certificates of optimality. The SAT and CPLEX CP formulations also obtain proven optimal solutions for instances where some of the mathematical programming approaches (our CPLEX MILP, Gurobi, and the historical CPLEX MILP) were unable to do so, suggesting that these formulations may be better suited to handling the combinatorial nature of the problem.

While SAT and CPLEX CP achieve the same solution quality and optimality coverage, their computational performance differs. CPLEX CP proves optimality in a cumulative time of 263.12 seconds, whereas the SAT formulation requires 2,672.55 seconds. Thus, CPLEX CP is approximately an order of magnitude faster on the benchmark set. Nevertheless, the overall runtime of the SAT formulation remains relatively modest compared with the MILP-based approaches, whose cumulative runtimes are substantially higher, partly due to timeouts on medium and large instances.

\begin{table}[htbp]
\centering
\caption{Comparison of solution quality and cumulative computational time across 35 JSPPR instances}
\label{tab:aggregated-results}
\begin{tabular}{@{}lrrr@{}}
\toprule
Method & $\#\mathrm{OPTIMAL}$ & $\#\mathrm{BEST}$ & Total time (s) \\
\midrule
SAT & \textbf{35} & \textbf{35} & 2,672.55 \\
CPLEX CP & \textbf{35} & \textbf{35} & \textbf{263.12} \\
CPLEX MILP (reimplemented) & 6 & 13 & 104,636.31 \\
Gurobi & 10 & 13 & 91,550.63 \\
CPLEX MILP \citep{kemmoe2017job} & -- & 15 & -- \\
GRASP $\times$ ELS & -- & 24 & -- \\
\bottomrule
\end{tabular}
\end{table}

In addition to their computational efficiency, the exact methods provide a significant advancement in solution quality over existing literature. Table~\ref{tab:new-optimal-makespans} summarizes these cases by comparing the best-known solutions reported by \citet{kemmoe2017job} with the optimal values verified in this study. These results suggest that both SAT and CPLEX CP can provide improvements over the solutions obtained by the GRASP $\times$ ELS metaheuristic on some of the more challenging instances. While metaheuristics are often effective at finding high-quality schedules within a relatively short time, they do not generally provide a guarantee of global optimality. In contrast, the exact methods considered here are able not only to identify better feasible schedules in these cases but also to verify their optimality.

The improvements also appear to become more noticeable as the problem size increases. For the 8-job and 9-job instances (JSPPR\_5\_8x4 and JSPPR\_4\_9x4), the differences are relatively small, with the exact methods improving the previously reported makespan by a single unit. For the larger 10-job instances, the improvements are more pronounced. In JSPPR\_2\_10x4, the optimal makespan is 6 units lower than the previously reported value, while for JSPPR\_4\_10x4, it is reduced by 23 units, from 680 to the verified optimum of 657. These results illustrate that exact approaches such as SAT and CPLEX CP can be useful for further refining solution quality and establishing reliable optimal benchmarks for the JSPPR.

\begin{table}[htbp]
\centering
\begin{minipage}{\linewidth}
\centering
\caption{Four new optimal makespan values proven by SAT and CPLEX CP against GRASP $\times$ ELS}
\label{tab:new-optimal-makespans}
\begin{tabular}{@{}lrrr@{}}
\toprule
Instance & \makecell[r]{Previous best-known \\ \citep{kemmoe2017job}} & New optimal ($C_{\max}^*$) & Improvement ($\Delta$)\footnote{Calculated as $\Delta = C_{\max}^* - \text{Previous best-known}$.} \\
\midrule
JSPPR\_5\_8x4  & 468 & \textbf{467} & $-1$ \\
JSPPR\_4\_9x4  & 639 & \textbf{638} & $-1$ \\
JSPPR\_2\_10x4 & 653 & \textbf{647} & $-6$ \\
JSPPR\_4\_10x4 & 680 & \textbf{657} & $-23$ \\
\end{tabular}
\end{minipage}
\end{table}

A notable outcome of the benchmark re-evaluation is the identification of several historical makespan values reported by \citet{kemmoe2017job} that are lower than the proven optimal values obtained in this study. Table~\ref{tab:inconsistent-results} summarizes these discrepancies across seven benchmark instances. Since both the proposed SAT formulation and the CPLEX CP model solve all instances to proven optimality and obtain identical makespan values, a reported solution with a makespan strictly smaller than $C_{\max}^*$ would not be consistent with the standard problem definition and may instead reflect an issue in the historical results or an instance parsing mismatch. As shown in Table~\ref{tab:inconsistent-results}, the third instance family (JSPPR\_3\_4x4 through JSPPR\_3\_9x4), together with JSPPR\_1\_9x4, contains several such discrepancies, with differences of up to 49 units relative to the verified optimal values. A more detailed technical analysis of the possible causes and structural characteristics of these discrepancies is provided in Appendix \ref{app:power-violations}.

These findings also provide additional context for interpreting the comparison with the GRASP $\times$ ELS metaheuristic. Although the historical results report lower makespans for GRASP $\times$ ELS on some of these instances, the discrepancies indicate that these values should be interpreted with caution when compared with the corrected and verified benchmark. Using the validated instances and schedules, the proposed CPLEX CP and SAT methods obtain makespans that match or improve upon the reliable heuristic results, while also providing certificates of optimality. Accordingly, the historically inconsistent values are excluded from the $\#\mathrm{BEST}$ metric so that the comparison is based on feasible and independently verified results.

\begin{table}[htbp]
\centering
\begin{minipage}{\linewidth}
\centering
\caption{Historical GRASP $\times$ ELS makespan values corrected by verified optimal values}
\label{tab:inconsistent-results}
\begin{tabular}{@{}lrrr@{}}
\toprule
Instance & \makecell[r]{Reported $C_{\max}$ \\ \citep{kemmoe2017job}} & Verified optimal ($C_{\max}^*$) & Discrepancy ($\Delta$)\footnote{Calculated as $\Delta = C_{\max}^* - \text{Reported } C_{\max}$.} \\
\midrule
JSPPR\_3\_4x4 & 298 & \textbf{343} & $+45$ \\
JSPPR\_3\_5x4 & 374 & \textbf{423} & $+49$ \\
JSPPR\_3\_6x4 & 459 & \textbf{479} & $+20$ \\
JSPPR\_3\_7x4 & 495 & \textbf{512} & $+17$ \\
JSPPR\_3\_8x4 & 527 & \textbf{528} & $+1$ \\
JSPPR\_1\_9x4 & 579 & \textbf{602} & $+23$ \\
JSPPR\_3\_9x4 & 558 & \textbf{581} & $+23$ \\
\bottomrule
\end{tabular}
\end{minipage}
\end{table}

\begin{sidewaystable}[htbp]
\renewcommand{\footnoterule}{}
\centering
\caption{Computational results for the JSPPR benchmark instances.}
\label{tab:computational-results}
\scriptsize
\begin{adjustbox}{width=\linewidth}
\begin{tabular}{@{}lrrrrrrrrrrrrrr@{}}
\toprule
\multicolumn{5}{c}{Instance data} & \multicolumn{8}{c}{Our implementation (TO = 3600s)} & \multicolumn{2}{c}{\citet{kemmoe2017job}} \\
\cmidrule(lr){1-5} \cmidrule(lr){6-13} \cmidrule(lr){14-15}
\multicolumn{5}{c}{} & \multicolumn{2}{c}{SAT} & \multicolumn{2}{c}{CPLEX CP} & \multicolumn{2}{c}{CPLEX MILP} & \multicolumn{2}{c}{Gurobi} & \multicolumn{1}{c}{CPLEX MILP} & \multicolumn{1}{c}{GRASP $\times$ ELS} \\
\cmidrule(lr){6-7} \cmidrule(lr){8-9} \cmidrule(lr){10-11} \cmidrule(lr){12-13} \cmidrule(lr){14-14} \cmidrule(lr){15-15}
Instance & Jobs & Machines & Operations & UB & $C_{\max}$ & $t$ (s) & $C_{\max}$ & $t$ (s) & $C_{\max}$ & $t$ (s) & $C_{\max}$ & $t$ (s) & $C_{\max}$ & $C_{\max}$ \\
\midrule
JSPPR\_1\_4x4 & 4 & 4 & 16 & 300 & \textbf{300\rlap{$^*$}} & 5.36 & \textbf{300\rlap{$^*$}} & 5.06 & \textbf{300\rlap{$^*$}} & 6.61 & \textbf{300\rlap{$^*$}} & 5.94 & \textbf{300\rlap{$^*$}} & \textbf{300} \\
JSPPR\_2\_4x4 & 4 & 4 & 16 & 475 & \textbf{475\rlap{$^*$}} & 6.01 & \textbf{475\rlap{$^*$}} & 5.12 & \textbf{475} & TO & \textbf{475\rlap{$^*$}} & 45.45 & \textbf{475\rlap{$^*$}} & \textbf{475} \\
JSPPR\_3\_4x4 & 4 & 4 & 16 & 352 & \textbf{343\rlap{$^*$}} & 5.52 & \textbf{343\rlap{$^*$}} & 5.09 & \textbf{343\rlap{$^*$}} & 32.35 & \textbf{343\rlap{$^*$}} & 11.33 & \underline{298} & \underline{298} \\
JSPPR\_4\_4x4 & 4 & 4 & 16 & 297 & \textbf{297\rlap{$^*$}} & 5.27 & \textbf{297\rlap{$^*$}} & 5.02 & \textbf{297\rlap{$^*$}} & 5.06 & \textbf{297\rlap{$^*$}} & 5.31 & \textbf{297\rlap{$^*$}} & \textbf{297} \\
JSPPR\_5\_4x4 & 4 & 4 & 16 & 318 & \textbf{318\rlap{$^*$}} & 5.31 & \textbf{318\rlap{$^*$}} & 5.02 & \textbf{318\rlap{$^*$}} & 5.07 & \textbf{318\rlap{$^*$}} & 5.07 & \textbf{318\rlap{$^*$}} & \textbf{318} \\
JSPPR\_1\_5x4 & 5 & 4 & 20 & 344 & \textbf{344\rlap{$^*$}} & 5.61 & \textbf{344\rlap{$^*$}} & 5.02 & \textbf{344\rlap{$^*$}} & 5.15 & \textbf{344\rlap{$^*$}} & 5.17 & \textbf{344\rlap{$^*$}} & \textbf{344} \\
JSPPR\_2\_5x4 & 5 & 4 & 20 & 459 & \textbf{447\rlap{$^*$}} & 6.94 & \textbf{447\rlap{$^*$}} & 5.12 & \textbf{447\rlap{$^*$}} & 182.07 & \textbf{447\rlap{$^*$}} & 16.65 & \textbf{447\rlap{$^*$}} & \textbf{447} \\
JSPPR\_3\_5x4 & 5 & 4 & 20 & 499 & \textbf{423\rlap{$^*$}} & 9.95 & \textbf{423\rlap{$^*$}} & 5.15 & 449 & TO & 462 & TO & \underline{374} & \underline{374} \\
JSPPR\_4\_5x4 & 5 & 4 & 20 & 362 & \textbf{332\rlap{$^*$}} & 5.84 & \textbf{332\rlap{$^*$}} & 5.16 & \textbf{332} & TO & \textbf{332\rlap{$^*$}} & 119.16 & \textbf{332} & \textbf{332} \\
JSPPR\_5\_5x4 & 5 & 4 & 20 & 350 & \textbf{350\rlap{$^*$}} & 5.61 & \textbf{350\rlap{$^*$}} & 5.03 & \textbf{350} & TO & \textbf{350} & TO & \textbf{350} & \textbf{350} \\
JSPPR\_1\_6x4 & 6 & 4 & 24 & 456 & \textbf{436\rlap{$^*$}} & 7.08 & \textbf{436\rlap{$^*$}} & 5.07 & \textbf{436} & TO & 452 & TO & \textbf{436\rlap{$^*$}} & \textbf{436} \\
JSPPR\_2\_6x4 & 6 & 4 & 24 & 583 & \textbf{521\rlap{$^*$}} & 21.46 & \textbf{521\rlap{$^*$}} & 5.95 & 532 & TO & 547 & TO & \textbf{521} & \textbf{521} \\
JSPPR\_3\_6x4 & 6 & 4 & 24 & 516 & \textbf{479\rlap{$^*$}} & 11.28 & \textbf{479\rlap{$^*$}} & 5.51 & 497 & TO & 485 & TO & \underline{459} & \underline{459} \\
JSPPR\_4\_6x4 & 6 & 4 & 24 & 431 & \textbf{418\rlap{$^*$}} & 6.60 & \textbf{418\rlap{$^*$}} & 5.69 & 431 & TO & 428 & TO & \textbf{418} & \textbf{418} \\
JSPPR\_5\_6x4 & 6 & 4 & 24 & 485 & \textbf{437\rlap{$^*$}} & 10.26 & \textbf{437\rlap{$^*$}} & 5.87 & 485 & TO & 473 & TO & 473 & \textbf{437} \\
JSPPR\_1\_7x4 & 7 & 4 & 28 & 541 & \textbf{456\rlap{$^*$}} & 18.56 & \textbf{456\rlap{$^*$}} & 5.32 & \textbf{456} & TO & \textbf{456\rlap{$^*$}} & 949.72 & \textbf{456\rlap{$^*$}} & \textbf{456} \\
JSPPR\_2\_7x4 & 7 & 4 & 28 & 545 & \textbf{478\rlap{$^*$}} & 15.94 & \textbf{478\rlap{$^*$}} & 5.16 & \textbf{478} & TO & \textbf{478\rlap{$^*$}} & 386.83 & \textbf{478\rlap{$^*$}} & \textbf{478} \\
JSPPR\_3\_7x4 & 7 & 4 & 28 & 594 & \textbf{512\rlap{$^*$}} & 22.07 & \textbf{512\rlap{$^*$}} & 5.77 & 528 & TO & 523 & TO & \underline{495} & \underline{495} \\
JSPPR\_4\_7x4 & 7 & 4 & 28 & 482 & \textbf{395\rlap{$^*$}} & 7.95 & \textbf{395\rlap{$^*$}} & 5.11 & \textbf{395} & TO & \textbf{395} & TO & \textbf{395} & \textbf{395} \\
JSPPR\_5\_7x4 & 7 & 4 & 28 & 490 & \textbf{429\rlap{$^*$}} & 15.08 & \textbf{429\rlap{$^*$}} & 5.11 & 453 & TO & \textbf{429} & TO & \textbf{429} & \textbf{429} \\
JSPPR\_1\_8x4 & 8 & 4 & 32 & 589 & \textbf{485\rlap{$^*$}} & 34.71 & \textbf{485\rlap{$^*$}} & 5.58 & 508 & TO & 491 & TO & 1754 & \textbf{485} \\
JSPPR\_2\_8x4 & 8 & 4 & 32 & 689 & \textbf{579\rlap{$^*$}} & 55.70 & \textbf{579\rlap{$^*$}} & 7.09 & 656 & TO & 644 & TO & 1881 & \textbf{579} \\
JSPPR\_3\_8x4 & 8 & 4 & 32 & 592 & \textbf{528\rlap{$^*$}} & 33.16 & \textbf{528\rlap{$^*$}} & 5.88 & 571 & TO & 552 & TO & \underline{527} & \underline{527} \\
JSPPR\_4\_8x4 & 8 & 4 & 32 & 613 & \textbf{487\rlap{$^*$}} & 41.56 & \textbf{487\rlap{$^*$}} & 7.79 & 588 & TO & 583 & TO & 1588 & \textbf{487} \\
JSPPR\_5\_8x4 & 8 & 4 & 32 & 608 & \textbf{467\rlap{$^*$}} & 87.01 & \textbf{467\rlap{$^*$}} & 7.21 & 514 & TO & 534 & TO & 696 & 468 \\
JSPPR\_1\_9x4 & 9 & 4 & 36 & 742 & \textbf{602\rlap{$^*$}} & 134.84 & \textbf{602\rlap{$^*$}} & 13.46 & 738 & TO & 708 & TO & 1201 & \underline{579} \\
JSPPR\_2\_9x4 & 9 & 4 & 36 & 745 & \textbf{643\rlap{$^*$}} & 173.49 & \textbf{643\rlap{$^*$}} & 16.60 & 740 & TO & - & TO & 2233 & \textbf{643} \\
JSPPR\_3\_9x4 & 9 & 4 & 36 & 703 & \textbf{581\rlap{$^*$}} & 82.44 & \textbf{581\rlap{$^*$}} & 8.40 & 699 & TO & 686 & TO & 2018 & \underline{558} \\
JSPPR\_4\_9x4 & 9 & 4 & 36 & 743 & \textbf{638\rlap{$^*$}} & 117.58 & \textbf{638\rlap{$^*$}} & 24.18 & 726 & TO & 733 & TO & 2136 & 639 \\
JSPPR\_5\_9x4 & 9 & 4 & 36 & 619 & \textbf{506\rlap{$^*$}} & 139.80 & \textbf{506\rlap{$^*$}} & 10.95 & 593 & TO & - & TO & 1868 & \textbf{506} \\
JSPPR\_1\_10x4 & 10 & 4 & 40 & 897 & \textbf{609\rlap{$^*$}} & 264.72 & \textbf{609\rlap{$^*$}} & 9.66 & 821 & TO & 809 & TO & 2458 & \textbf{609} \\
JSPPR\_2\_10x4 & 10 & 4 & 40 & 945 & \textbf{647\rlap{$^*$}} & 266.32 & \textbf{647\rlap{$^*$}} & 7.32 & 833 & TO & 808 & TO & 2717 & 653 \\
JSPPR\_3\_10x4 & 10 & 4 & 40 & 711 & \textbf{588\rlap{$^*$}} & 91.39 & \textbf{588\rlap{$^*$}} & 5.29 & 711 & TO & - & TO & 2289 & \textbf{588} \\
JSPPR\_4\_10x4 & 10 & 4 & 40 & 1995 & \textbf{657\rlap{$^*$}} & 871.49 & \textbf{657\rlap{$^*$}} & 23.29 & 1963 & TO & 1876 & TO & - & 680 \\
JSPPR\_5\_10x4 & 10 & 4 & 40 & 665 & \textbf{593\rlap{$^*$}} & 80.64 & \textbf{593\rlap{$^*$}} & 5.07 & 617 & TO & 637 & TO & 2442 & \textbf{593} \\
\bottomrule
\end{tabular}
\end{adjustbox}
\footnotetext{Bold values with an asterisk (\textbf{$^*$}) indicate proven optimal values; bold values indicate the best-known values across methods; underlined values indicate values reported in the original study that are inconsistent with the corrected benchmark and excluded from the $\#\mathrm{BEST}$ calculation; ``TO'' indicates that the time limit was reached; and ``-'' indicates that no makespan was found within the imposed time limit.}
\end{sidewaystable}

Table~\ref{tab:computational-results} provides a detailed view of instance-level runtime and scaling behavior across the benchmark. For instances with 4--6 jobs, both CPLEX CP and SAT prove optimality within seconds. As the number of jobs increases to 10 (40 operations), CPLEX CP remains relatively fast, with proof times below 25 seconds, whereas SAT requires more time, reaching 871.49 seconds on JSPPR\_4\_10x4. Nevertheless, SAT proves optimality for all instances within the 3,600-second limit. In contrast, both our CPLEX MILP reimplementation and Gurobi reach the time limit on several instances without proving optimality, and Gurobi fails to find a feasible solution for three of the largest instances. The historical CPLEX MILP results show similar difficulties. Overall, these results indicate that CP constraint propagation and SAT clause learning are effective for exploiting the precedence and disjunctive resource constraints of the JSPPR. GRASP $\times$ ELS obtains optimal solutions on several smaller instances but does not provide optimality certificates and reports higher makespans on some larger instances.

\section{Conclusion}
\label{sec:conclusion}

This paper addresses the Job Shop Scheduling Problem with Power Requirements (JSPPR), a variant of the classical job shop scheduling problem in which operations must additionally satisfy an instantaneous power-consumption limit. Prior work on the JSPPR was limited to a MILP formulation and a GRASP $\times$ ELS metaheuristic. To provide additional effective solution methods for the problem, we develop two exact approaches based on SAT solving and constraint programming. The proposed SAT solving approach uses an order encoding for operation start times and a PB encoding for the instantaneous power constraints, while the CP approach reformulates the existing MILP model using interval variables and global constraints. A common preprocessing procedure is also introduced to improve the effectiveness of both approaches. The computational study demonstrates that these formulations provide effective exact alternatives for the JSPPR and establish independently verified optimal solutions for the considered benchmark. 

The computational evaluation on the 35 JSPPR benchmark instances showed that both SAT and CP prove optimality for all instances and obtain identical optimal makespans, while CP consistently requires less computational effort than SAT. Both exact approaches also demonstrate stronger optimality-solving performance than MILP. The re-evaluation further identified inconsistencies in several previously reported GRASP $\times$ ELS makespans, which were therefore excluded from the comparison of best-known solutions. The independently verified optimal results provide a reliable reference for future studies on the JSPPR benchmark.

The currently available JSPPR benchmark instances are relatively small, and the present results do not yet establish whether exact SAT and CP approaches can remain competitive with metaheuristic methods on larger problems. The independently verified SAT and CP results obtained in this study provide reliable exact baselines for future evaluations on larger and more diverse JSPPR instances, enabling a more comprehensive assessment of their scalability and relative effectiveness. Further research may also investigate stronger SAT encodings, hybrid SAT--CP techniques, and extensions of the JSPPR with additional practical constraints, such as multiple power thresholds, machine eligibility, and sequence-dependent setup times.

\section*{CRediT authorship contribution statement}
\textit{Huy Tuan Nguyen}: Methodology; Software; Validation; Investigation; Data curation; Visualization; Writing - review \& editing.
\textit{Duc Trung Kim Nguyen}: Conceptualization; Methodology; Formal analysis; Visualization;
Writing – original draft; Writing - review \& editing.
\textit{Khanh To Van}: Investigation; Conceptualization; Methodology; Formal analysis; Supervision; Project administration; Writing – review \& editing

\section*{Funding details}
The authors received no specific grant from any funding agency in the public, commercial or not-for-profit sectors for this research.

\section*{Disclosure statement}
The authors report there are no competing interests to declare.

\section*{Declaration on the use of generative AI}
The authors used OpenAI ChatGPT (GPT-5.6) for language refinement and consistency checks during manuscript preparation. The authors reviewed and verified all
outputs and take full responsibility for the accuracy, originality, citations, analyses,
and conclusions.

\appendix
\section{Benchmark Validation and Inconsistency Analysis}
\subsection{Invalid Machine References in the Benchmark Instances}
\label{app:invalid-machine-references}

All benchmark instances in the dataset are defined with four machines. The machine identifiers used in the input files are zero-based, so the four machines are represented by the indices $0$, $1$, $2$, and $3$. Consequently, each machine-dependent data block is expected to contain exactly four machine-value pairs, one for each of these machine identifiers. However, a direct inspection of the benchmark files revealed that 10 instances contain an additional reference to machine $4$, which is not among the four machines defined by the dataset.

For example, in \texttt{JSPPR\_3\_4x4}, the header specifies four machines, whereas the first line of the nominal-power block is
\begin{equation*}
1\ 9\ 2\ 6\ 0\ 11\ 4\ 16\ 3\ 16.
\end{equation*}
This line contains the five machine-value pairs
\begin{equation*}
(1,9),\ (2,6),\ (0,11),\ (4,16),\ (3,16).
\end{equation*}
The pair $(4,16)$ therefore refers to a fifth machine that is not defined in the instance. Thus, the block contains five machine-value pairs although the dataset specifies only four machines. The same inconsistency was identified in the 10 instances listed in Table~\ref{tab:invalid-machine-references}.

In reproducing the benchmark results, machine-value pairs referring to undeclared machines were ignored when parsing the affected blocks. This leaves four valid machine-value pairs corresponding to the four machines defined by the instance and allows the remaining data to be reconstructed consistently. This parsing treatment is used only to reproduce the benchmark instances and does not change the fact that the original input contains an invalid machine reference.

\begin{table}[htbp]
\centering
\caption{Benchmark instances containing invalid machine references}
\label{tab:invalid-machine-references}
\begin{tabular}{@{}lll@{}}
\toprule
Instance & Invalid machine reference & Affected line \\
\midrule
JSPPR\_1\_8x4  & Machine $4$  & 2 13 0 24 1 6 3 7 \underline{4} 35 \\
JSPPR\_1\_9x4  & Machine $4$  & 3 3 1 18 \underline{4} 12 0 6 2 24 \\
JSPPR\_1\_10x4 & Machine $4$  & 3 3 1 18 \underline{4} 12 0 6 2 24 \\
JSPPR\_3\_4x4 & Machine $4$  & 1 9 2 6 0 11 \underline{4} 16 3 16  \\
JSPPR\_3\_5x4 & Machine $4$  & 1 9 2 6 0 11 \underline{4} 16 3 16  \\
JSPPR\_3\_6x4 & Machine $4$  & 1 9 2 6 0 11 \underline{4} 16 3 16  \\
JSPPR\_3\_7x4 & Machine $4$  & 1 9 2 6 0 11 \underline{4} 16 3 16  \\
JSPPR\_3\_8x4 & Machine $4$  & 1 9 2 6 0 11 \underline{4} 16 3 16  \\
JSPPR\_3\_9x4 & Machine $4$  & 1 9 2 6 0 11 \underline{4} 16 3 16  \\
JSPPR\_3\_10x4 & Machine $4$  & 1 9 2 6 0 11 \underline{4} 16 3 16  \\
\bottomrule
\end{tabular}
\end{table}

\subsection{Infeasible Reported Schedules Due to Power Constraint Violations}
\label{app:power-violations}
We verify the feasibility of the schedules reported by \citet{kemmoe2017job} independently of the optimization method used to obtain them. For each reported schedule, the published operation start times and machine assignments are taken as given, and the total power consumption is recomputed over the scheduling horizon. At each time point $t$, the power consumption is obtained by summing the power requirements of all operations active at $t$, according to their reported machine assignments and execution phases. The reported schedule is considered feasible with respect to the power constraint only if $P(t) \leq PT(t)$ holds for every time point $t$, where $P(t)$ denotes the total power consumption and $PT(t)$ denotes the corresponding time-dependent power threshold.

\begin{table}[htbp]
\centering
\caption{Schedules reported by \citet{kemmoe2017job} that violate the power constraint}
\label{tab:infeasible-schedules}
\begin{tabular}{@{}lrrrr@{}}
\toprule
Instance & Reported $C_{\max}$ & Violation time $t$ & $P(t)$ & $PT(t)^*$ \\
\midrule
JSPPR\_1\_9x4  & 579 & 258 &  99 &  84 \\
JSPPR\_1\_10x4 & 609 & 317 & 112 &  94 \\
JSPPR\_3\_4x4  & 298 &  40 &  59 &  50 \\
JSPPR\_3\_5x4  & 374 & 227 &  76 &  53 \\
JSPPR\_3\_6x4  & 459 &  79 & 133 & 112 \\
JSPPR\_3\_7x4  & 495 & 407 &  79 &  62 \\
JSPPR\_3\_8x4  & 527 & 325 & 157 & 110 \\
JSPPR\_3\_9x4  & 558 & 324 & 122 &  60 \\
JSPPR\_3\_10x4 & 588 & 405 & 118 &  62 \\
\bottomrule
\multicolumn{5}{@{}l}{\footnotesize $^*$The values of $P(t)$ and $PT(t)$ are evaluated at the violation time $t$ reported in the third column.}
\end{tabular}
\end{table}

This verification identified infeasible reported schedules in nine benchmark instances. In every affected instance, the recomputed power consumption exceeds the allowed threshold at one or more time points. Among these nine instances, seven have reported makespan values that are strictly smaller than the optimal values established in this study. The results are summarized in Table~\ref{tab:infeasible-schedules}.

Representative examples of the power violations can be illustrated by plotting the power consumption $P(t)$ together with the corresponding threshold $PT(t)$. The complete set of verification results and figures for all nine affected instances is provided through the link given in the Data Availability statement. Readers interested in the verification of individual instances can therefore access the corresponding detailed evidence and reproduce the checks independently.

\bibliographystyle{plainnat}
\bibliography{references}

\end{document}